\documentclass[narrowdisplay]{amsart}

\usepackage{enumerate}

\usepackage{verbatim, amssymb,amsmath,amsthm,dsfont}
\usepackage{color,natbib,graphicx}

\theoremstyle{plain} 
\newtheorem{thm}{Theorem}
\newtheorem{lem}[thm]{Lemma} 

\newtheorem{cor}[thm]{Corollary}

\theoremstyle{definition} \newtheorem{defn}[thm]{Definition}
\newtheorem{rem}[thm]{Remark}

\newcommand{\PP}{\mathbb P}

\begin{document}

\title{There are no caterpillars in a wicked forest}

\date{\today}

\author{James H. Degnan}
\address{JHD: Department of Mathematics and Statistics, University of New Mexico
  Albuquerque, NM 87131, USA}
\email{jamdeg@unm.edu}

\author{John A. Rhodes}
 \address{JAR: Department of Mathematics and Statistics,
University of Alaska Fairbanks,
  PO Box 756660,
  Fairbanks, AK 99775,
 USA}
\email{j.rhodes@alaska.edu}

\maketitle

\bibliographystyle{natbib}

\begin{abstract}
Species trees represent the historical divergences of populations or species, while gene trees trace the ancestry of individual gene copies sampled within those populations. In cases involving rapid speciation, gene trees with topologies that differ from that of the species tree can be most probable under the standard multispecies coalescent model, making species tree inference more difficult.
Such {\it anomalous gene trees} are not well understood except for some small cases.  In this work, we establish one constraint that applies to trees of any size: gene trees with ``caterpillar" topologies cannot be anomalous.  The proof of this involves a new combinatorial object, called a {\it population history}, which keeps track of the number of coalescent events in each ancestral population.  \end{abstract}

Keywords: gene tree, species tree, multispecies coalescent, anomalous gene tree, coalescent history,  phylogeny

\section{Introduction}

An important distinction is made in phylogenetics between species trees and gene trees.  Species trees describe the ancestral relationships between populations of individuals (each carrying many genes) that have undergone divergences at various times in the past.  A gene tree tracks the ancestral relationships for a single gene sampled from individuals within extant species populations.  
In a species tree, the ancestral populations associated to edges have finite durations (see Figure \ref{fig:fig1}).  As a result, going backwards in time, several gene lineages from sampled individuals may remain distinct within a common ancestral population --- a phenomenon called \emph{incomplete lineage sorting} \citep{maddison1997} --- and then merge with other lineages to form a gene tree that is topologically dissimilar to the species tree. An understanding of this phenomenon, which
leads us to expect some, and possibly many, gene trees to differ from the species tree, is essential to statistical approaches to inference of species trees from genomic data sets.

\begin{figure*}
\begin{center}
\vspace{1in}
\includegraphics[width=.48\textwidth]{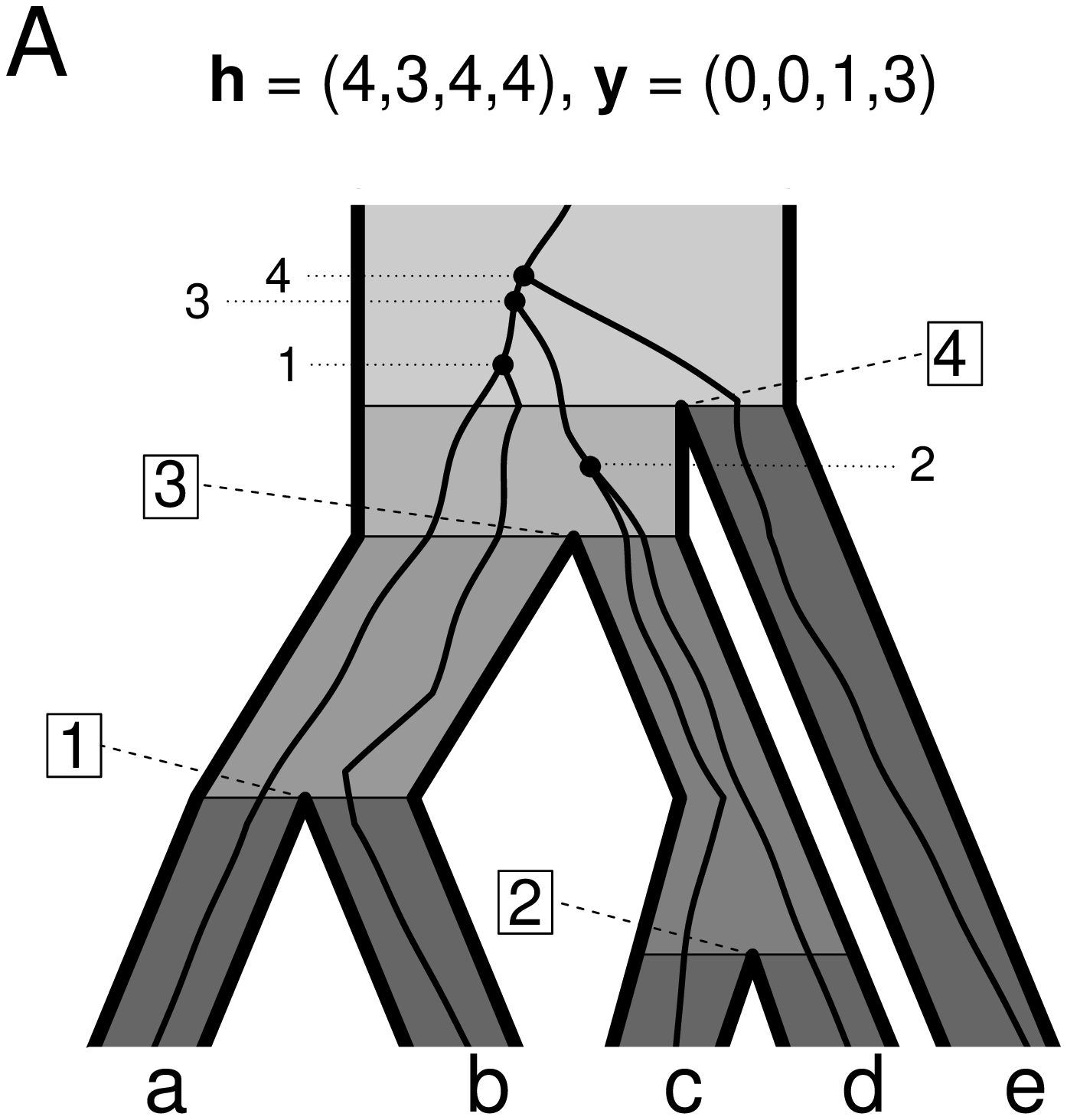}
\hspace{-1cm}\includegraphics[width=.48\textwidth]{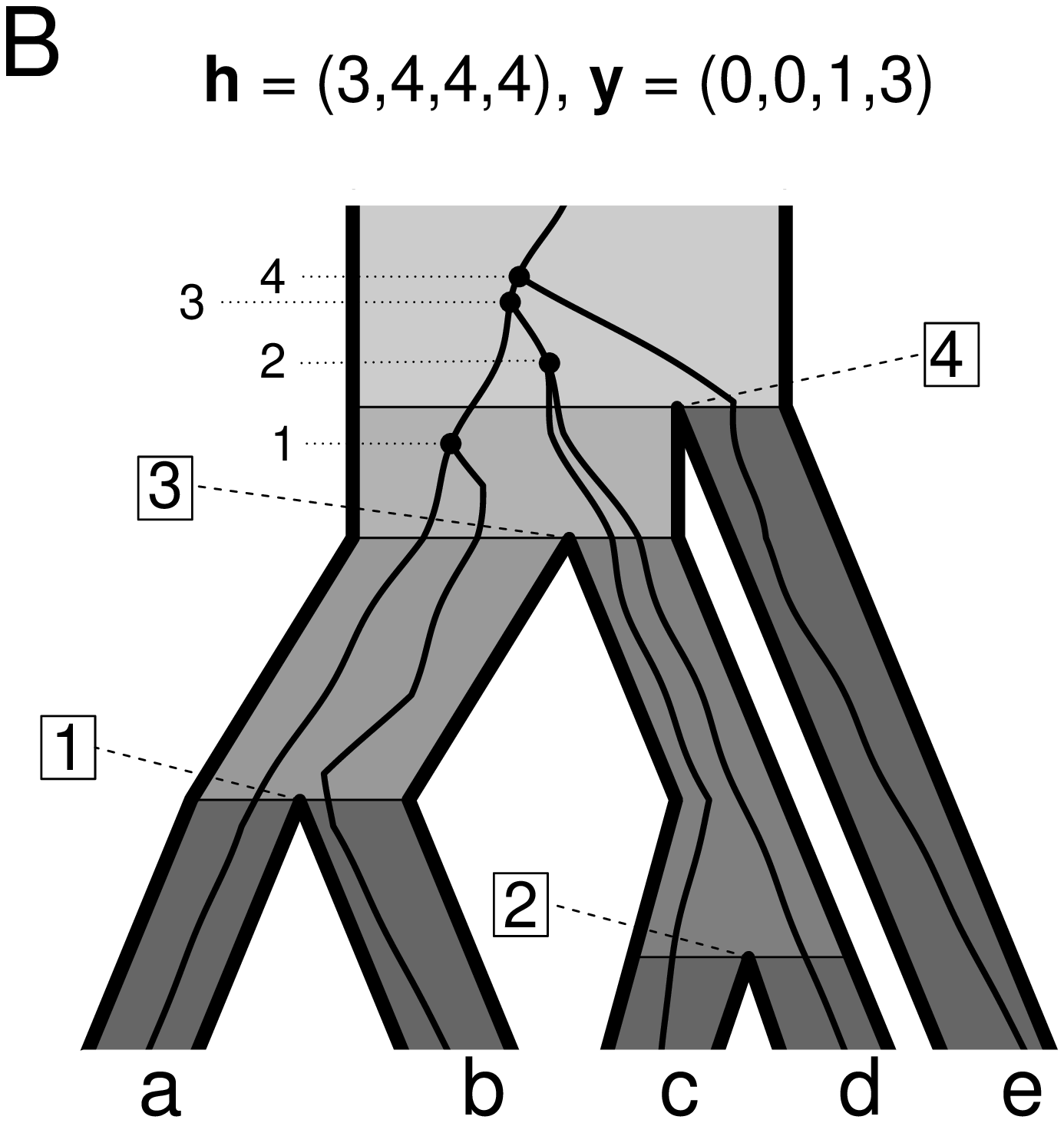}
\includegraphics[width=.48\textwidth]{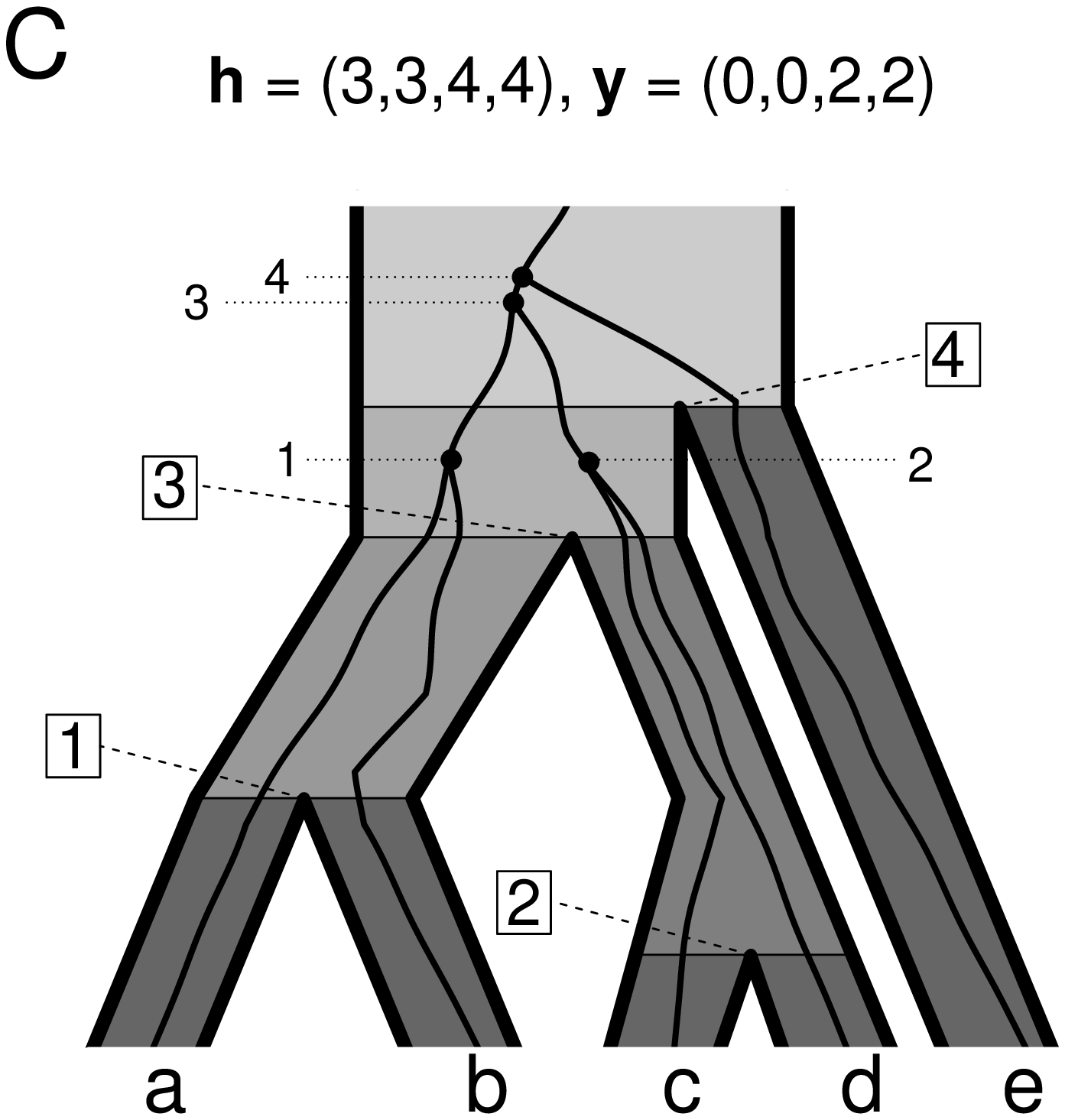}
\hspace{-1cm}\includegraphics[width=.48\textwidth]{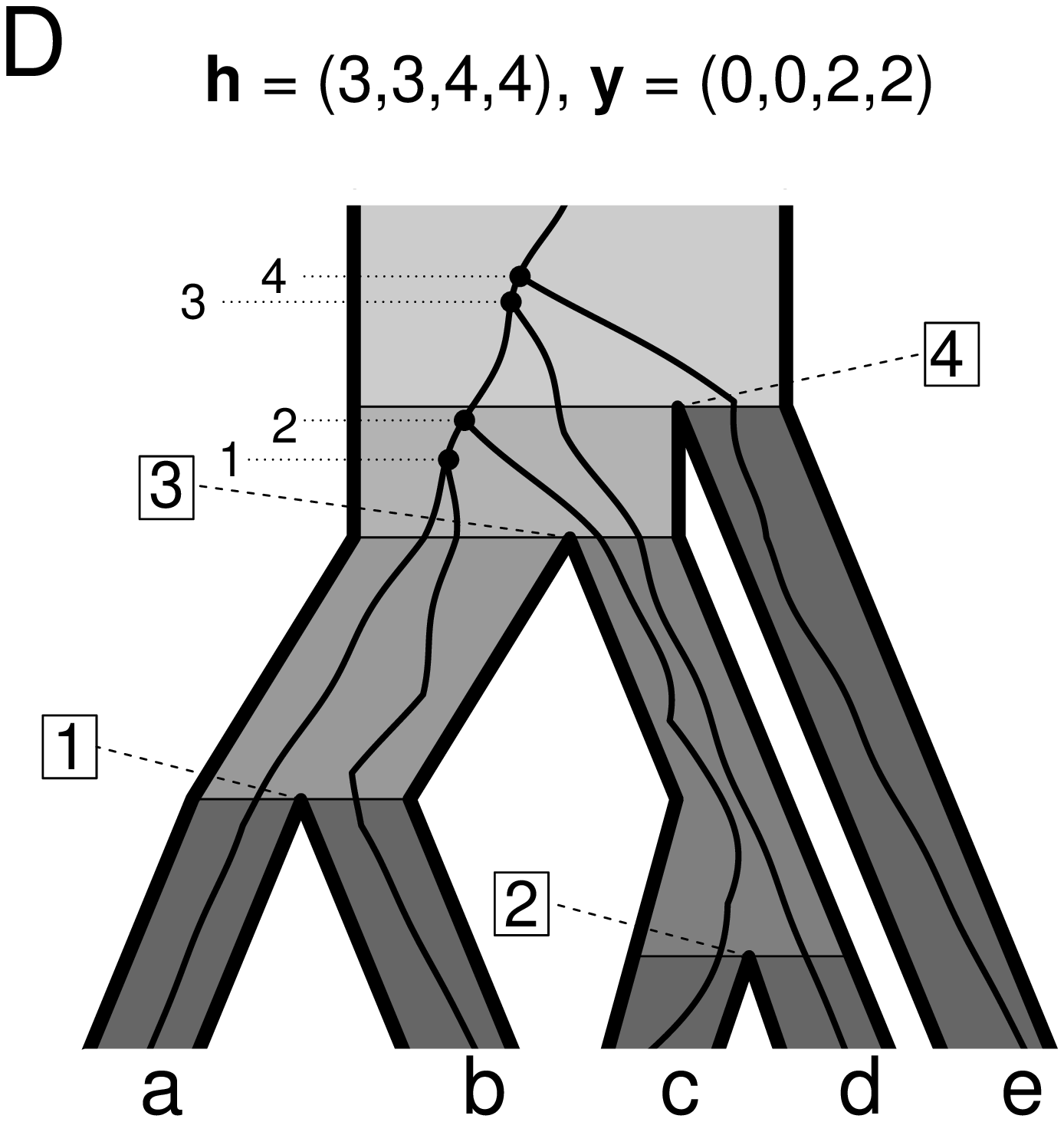}
\caption{A species tree with the matching gene tree (A,B,C) under three different coalescent histories (out of 13 possible), and a nonmatching caterpillar gene tree (D).  Speciation events occur when populations (shaded polygons) split into two new populations going forward in time (downward).  The population ancestral to the root of the species tree (lightest shading) is assumed to extend infinitely into the past; all other populations have finite durations.
The nodes of the trees are labelled in a postorder traversal using large, boxed numbers for the species tree, and unboxed numbers for the coalescent events. The vectors $\mathbf h, \mathbf y$ give coalescent histories and population histories, respectively, as explained in Section \ref{sec:defs}, using node labels as vector indices.}\label{fig:fig1}
\end{center}
\end{figure*}

The multispecies coalescent model gives a stochastic description of gene tree formation within a species tree. Kingman's coalescent model \citep{kingman1982a,hudson1983a,tajima1983,wakeley2008}  is adopted for each population (edge) of the the species tree,  so that the waiting time until coalescence between any pair of gene lineages within a population, going backwards in time, is exponentially distributed with mean 1.  At each node of the species tree, gene lineages reaching it from its descendent edges `enter' the population above starting a new coalescent process.
Combining calculations of probabilities for the within-population Kingman coalescent process with combinatorial features of the species tree, it is possible to calculate the probability of the formation of any topological gene tree
 \citep{degnan2005}.
A rooted species tree, with branch lengths, relating $n$ taxa thus determines a probability mass function on the set of all $(2n-3)!!$ rooted topological gene trees defined on the same species.

Under this model, the most likely gene tree topology does not necessarily match that of the species tree.
For example, the species tree $(((a,b),c),d)$, with choices of appropriate branch lengths, can result in any of the symmetric gene tree topologies, $((a,b),(c,d))$, $((a,c),(b,d))$, or $((a,d),(b,c))$, being more probable than the gene tree $(((a,b),c),d)$.  The term \emph{anomalous gene tree} (AGT) is used to describe gene trees that are more probable than the gene tree with the same topology as the species tree.  Although for four taxa, AGTs only arise for an asymmetric species tree, for any species tree topology with five or more taxa there are branch lengths (durations of internal populations) that lead to at least one AGT \citep{degnan2006}. 

Although this result describes the shapes of species trees that can have AGTs, less is known about gene tree shapes that can be AGTs.  For four taxa  \citep{degnan2006}, explicit computation of gene tree probabilities under the coalescent showed that only symmetric gene trees can be AGTs .  For five taxa \citep{rosenberg2008}, a computation showed that if the species tree is completely unbalanced, e.g, $((((a,b),c),d),e)$, then any gene tree with a different unlabeled topology can be an AGT. However, for five-taxon species trees of any topology, a completely unbalanced gene tree is never an AGT.  Furthermore, any noncaterpillar gene tree can be an AGT for some species tree. For example, if the species tree is a caterpillar, then any noncaterpillar gene tree is more probable than the matching gene tree if all species tree branch lengths are sufficiently short \citep{degnan2006}.

We refer to completely unbalanced trees, such as $((((a,b),c),d),e)$  and its analogs with more taxa, as \emph{rooted caterpillars}, usually omitting the word ``rooted" as this paper only concerns rooted trees.  We generalize the above observations by showing that for species trees of any size, there are no AGTs with caterpillar topologies.   This also implies the statement chosen as the title of this paper, using the terminology introduced in  \citet{degnan2006} which we restate in the next section.

While our results are theoretical, they have potential to contribute to the practice of species tree inference. For instance, when different genes yield different inferred phylogenetic trees, or different methods yield conflicting estimated species trees, evolutionary biologists sometimes wonder if their inferred tree is an AGT rather than the desired species tree \citep[e.g.][]{castillo-ramirez2008,zhaxybayeva2009}.  A recent paper uses a heuristic test based on taking subsets of four-taxa to conclude that there is evidence of the anomaly zone in a skink phylogeny \citep{linkem2014}.
One implication for our results is that if a phylogenetic method returns a caterpillar tree (as often happens in with smaller numbers of species), the empirical phylogeneticist can be sure that an AGT was not inferred.

\section{Notation and Definitions}\label{sec:defs}
Let $X$ denote a finite set, whose elements we refer to as \emph{taxa}. By a \emph{tree on $X$} we will mean a topological tree with leaves bijectively labeled by $X$.

\begin{defn}
A \emph{species tree $\sigma = (\psi,\boldsymbol{\lambda})$ on $X$} is a rooted, binary tree $\psi$ on $X$ together with a vector $\boldsymbol{\lambda} = (\lambda_{1}, \dots, \lambda_{n-2})$ of internal edge lengths (weights), where $n = |X|$, $\{e_1,\dots, e_{n-2}\}$ are the internal edges of $\psi$, and $\lambda_{i} > 0$  is the length of $e_i$ for $i=1, \dots, n-2$
\end{defn}

Nodes of the species tree represent speciation events, and edges represent populations extending over time. Edge lengths are given in coalescent units which (for constant population size) are the ratio of elapsed time to population size. It is convenient for the coalescent model to view $\psi$ as augmented by an additional directed edge leading to its root, in order to refer to a population ancestral to the root. We treat this edge as having infinite length, and consider it to be an internal edge of the species tree. 

The coalescent on a species tree $\sigma$ models the formation of gene trees by the merging of ancestral lineages (going backwards in time) within the populations represented by the tree's edges. We focus on the situation where one lineage is sampled per taxon, so pendant edge lengths for the species tree would be irrelevant. With this sampling scheme, a gene tree can also be leaf-labeled by $X$.

Since under the standard coalescent only binary gene trees have positive probability of being realized, and we are interested solely in the topological form of these trees, we make the following definition.

\begin{defn}
A \emph{gene tree}, $T$, on taxa $X$ is a rooted binary tree on $X$. 
\end{defn}

\begin{defn}
Given a species tree $\sigma = (\psi,\boldsymbol{\lambda})$, the \emph{matching gene tree} is the gene tree $T_M$ isomorphic to $\psi$ as a leaf-labeled tree.  
\end{defn}
 Though it is in some sense artificial to distinguish between $\psi$  and $T_M$, we do so in order to keep clear the difference in viewpoint between the fixed topological species tree $\psi$ and one of the possible states, $T_M$, of the gene tree random variable under the coalescent model.

Probabilities of an event $E$ under the 1-sample per taxon coalescent model on a species tree $\sigma$ are denoted $\PP_\sigma(E)$. In particular, the probability of a gene tree $T$ is $\PP_\sigma(T)$. (See Degnan and Salter, 2005,  for details on computations of such probabilities.)

\begin{defn}[Degnan and Rosenberg, 2006]
A gene tree $T$ is said to be an \emph{anomalous gene tree} (AGT) for a species tree $\sigma = (\psi,\boldsymbol{\lambda})$ if $\PP_{\sigma}(T) > \PP_{\sigma}(T_M)$.
\end{defn}

AGTs are significant, since their existence thwarts picking the most frequently occurring gene tree in a sample as the estimate of the species tree \citep{degnan2006}.  Though intuitively appealing, this democratic vote method is
not  statistically consistent. The following pathological situation is one where such voting is particularly misleading, in that voting based on gene trees arising from several species trees always ranks the true tree last.
\begin{defn}\citep{degnan2006}
A \emph{wicked forest} $W$ is a set of at least two species trees, with distinct topologies but defined on the same set of taxa $X$, such that for all $\sigma_i,\sigma_j  \in W$ with $i \ne j$, the gene tree $T^j_M$ matching $\sigma_j$  is an AGT for $\sigma_i$.
\end{defn}

The first set of trees noticed to form a wicked forest had six taxa and was given by \citet{degnan2006}. Their discovery was motivated by trying to find examples of trees that were AGTs yet were  less balanced than the matching tree.  Rosenberg and Tao (2008) fully characterized wicked forests for five-taxon trees, the smallest number of taxa for which wicked forests exist.   The maximum number of trees that can form a wicked forest for $n>5$ taxa is not known.  An example of a wicked forest with three trees is shown in Figure 2 and is based on swapping two-taxon clades in the trees. 

\begin{figure*}
 \setlength{\unitlength}{0.05 mm}%
\begin{picture}(45,0)(0,0)
\put(0,-100){\fontsize{12.23}{15.07}\selectfont   \makebox(0.0, 0.0)[c] {\;\;\;\;\;\;$I$\;\;\;\;\;\;\;\;\;\;\;\;\;\;\;\;\;\;\;\;\;\;\;\;\;\;\;\;\;\;\;\;\;$II$\;\;\;\;\;\;\;\;\;\;\;\;\;\;\;\;\;\;\;\;\;\;\;\;\;\;\;\;\;\;\;\;\;$III$}}
\end{picture}\\
\bigskip
\vspace{1.1cm}
\includegraphics[width=.32\textwidth]{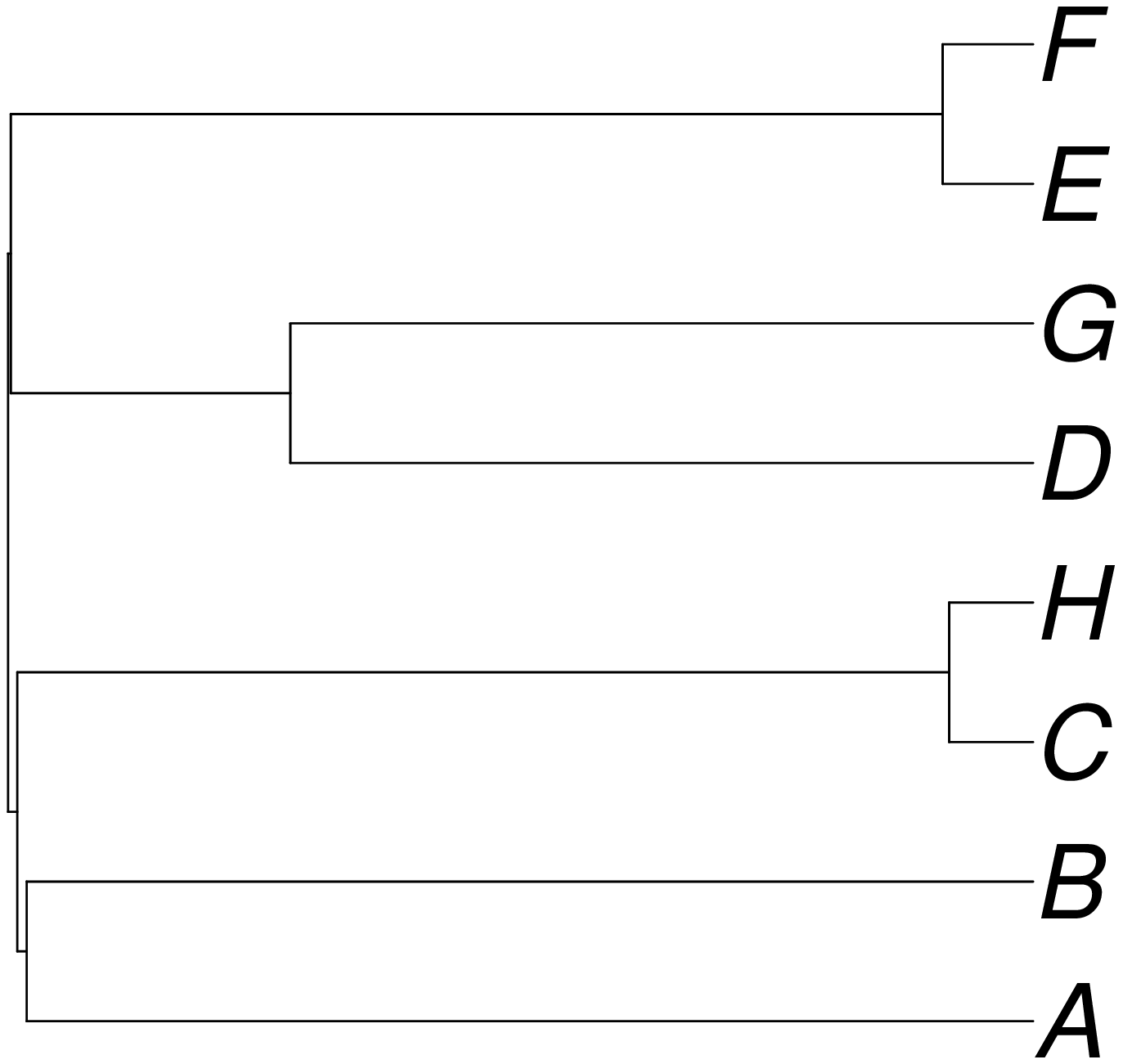}
\includegraphics[width=.32\textwidth]{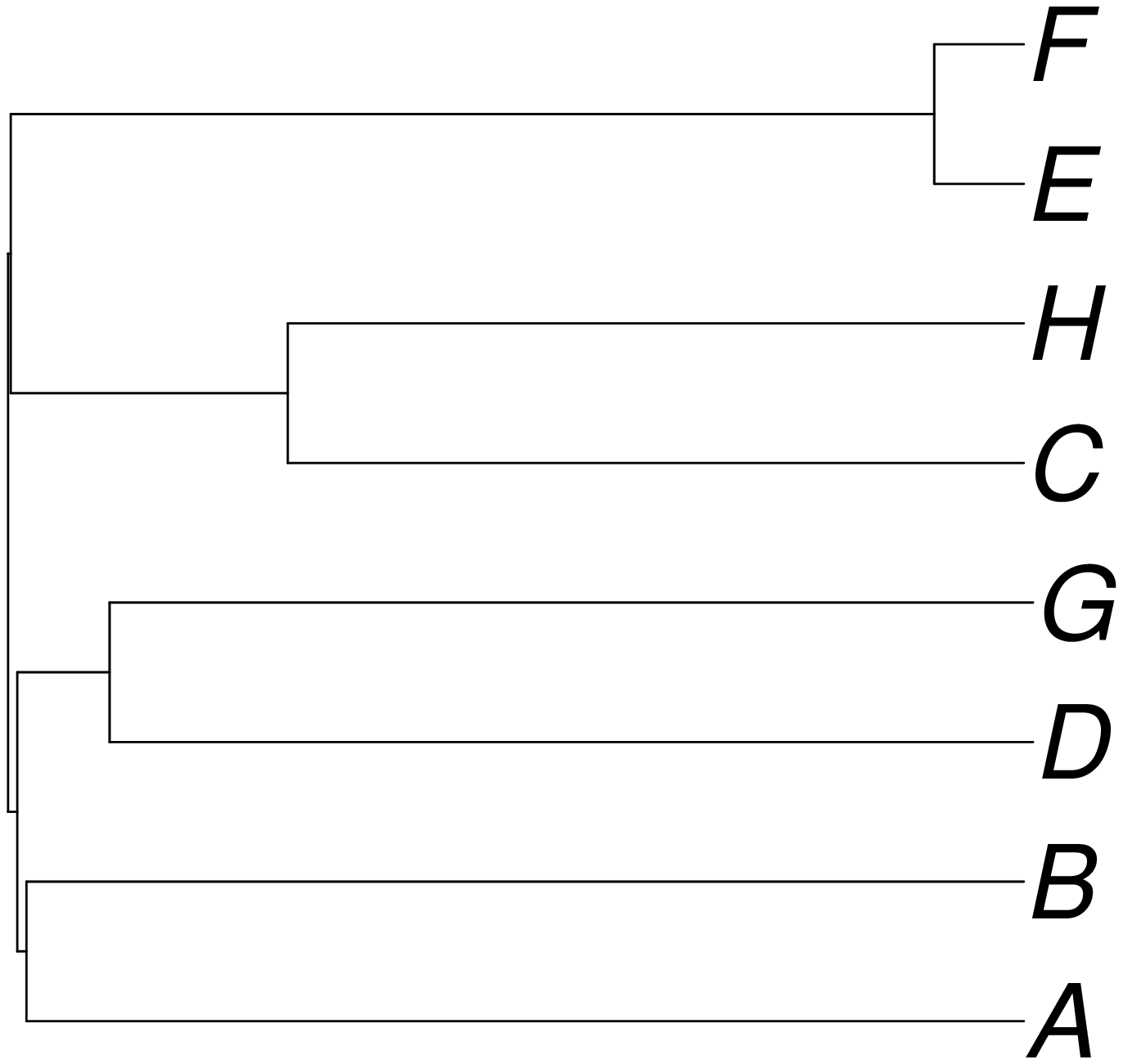}
\includegraphics[width=.32\textwidth]{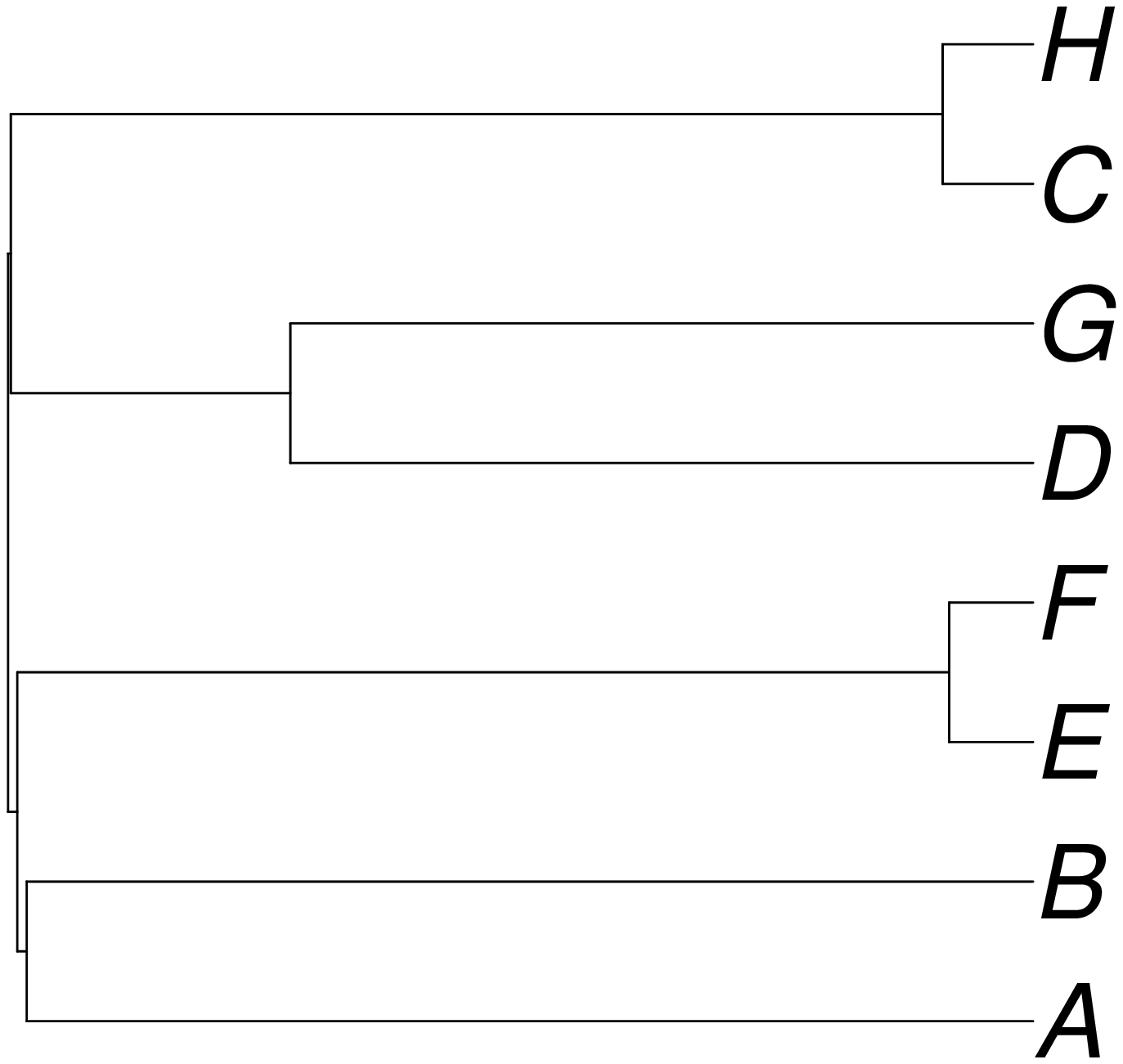}
\caption{A wicked forest with three balanced 8-taxon species trees.  Branch lengths are drawn to scale with the total depth of the tree equal to 0.11 coalescent units.  For each species tree $i \in \{ I, II, III\}$, the two gene trees with the matching topology for species tree $j \in \{I, II, III\} \smallsetminus \{i\}$ are AGTs for species tree $i$.  Species tree $I$ in Newick format is $(((A:.108,B:0.108):.001,(C:0.009,H:0.009):0.1):.0010,((D:0.0797,G:0.0797):0.03,(E:0.0097,F:0.0097):.100):.0003)$.}\label{F:wicked}
\end{figure*}  

\medskip

To compute and compare the probabilities of various gene trees under the coalescent model, we need further technical notions.

We treat all trees as directed graphs, with all edges directed away from the root (except, in species trees, for the ``edge'' ancestral to the root). Since we depict trees with the root placed above the leaves, we use terminology such as `ancestral' and `above,' or `descendent' and `below' interchangeably to describe directed relationships of nodes and edges.

Under the coalescent model on a species tree $\sigma=(\psi,\boldsymbol \lambda)$ on $X$, all gene trees $T$ on $X$ are realizable. That is, $\PP_{\sigma}(T)>0$ for all $T$. To compute $\PP_{\sigma}(T)$ one considers the various ways in which $T$ is realizable. This may be done at several levels of detail. The most detailed non-metric characterization would be to specify  \emph{coalescent histories with in-population rankings}, in which for each node of $T$ one indicates an edge of $\psi$ on which the coalescent event that node described occurred, as well as an ordering to the coalescent events within each species tree edge. (These are called {\it instantiations} of coalescent histories by \citet{degnan2005} ). 

A less detailed level is to specify \emph{coalescent histories}, where the ranking of coalescent events on edges is not recorded.
This is the key notion used by \citet{degnan2005} for the computation of gene tree probabilities (with adjustments for the count of possible in-population rankings). 

Finally, a \emph{population history} is an even cruder summary.  It records only the number of coalescent events on the edge, but does not record which lineages coalesced. To the best of our knowledge, this concept has not been used in previous works studying species trees and gene trees, though it plays an essential role in our arguments.

\smallskip

To formalize these notions, it is useful to encode the topology of a tree through the ancestral relationships of its nodes.
Let $V_T$ denote the set of nodes of a rooted tree $T$ (either a gene or species tree), and $I_T\subset V_T$ the subset of internal nodes.
Let $$\alpha_{ij} =\begin{cases}1&\text{ if node $i$ is ancestral or equal to node $j$,}\\0&\text{ otherwise.}\end{cases}$$
 This indicator function $\alpha$ on $V_T\times V_T$ fully encodes the topology of $T$. Labeling the edges of $T$ by the label of their end nodes, $\alpha$ also gives indicators of ancestral relationships between edges.

\begin{defn}\citep{degnan2005}\label{def:chist}
Let $\sigma = (\psi,\boldsymbol{\lambda})$ be a species tree and $T$  a gene tree on $X$, with $|X|=n$. A \emph{coalescent history for $T$} is an $(n-1)$-tuple $\mathbf{h} = \mathbf {h}_T= (h_i)_{i\in I_T}$ with each $h_i\in {I_\psi}$ satisfying
\begin{enumerate}
\item \label{cond:ch1} for all $i\in I_T$, the set of leaves descended from node $i$ of $T$ is a subset of the set of leaves descended from node $h_i$ of $\psi$; i.e., for all leaf labels $k$, $\alpha_{ik}=1$ implies $\alpha_{h_i k}=1$, and
\item \label{cond:ch2} if node $i$ is ancestral to node $j$ on $T$, then node $h_i$ is ancestral or equal to node $h_j$ on $\sigma$; i.e., $\alpha_{ij}=1$ implies $\alpha_{h_ih_j} = 1$.
\end{enumerate}
The set of coalescent histories for a species tree with topology $\psi$ and a gene tree $T$ is denoted $H_{\psi,T}$.
\end{defn}

Conceptually, such a history records that the coalescent event forming node $i$ of the gene tree occurs in the population immediately above node $h_i$ of the species tree. Condition (1) thus encodes the idea that coalescences must predate the most recent common ancestor of the populations from which they were sampled.  Condition (2) ensures that the sequence of coalescences is consistent with the topology of the gene tree; e.g., if a gene tree displays subtree $((a,b),c)$, then  $c$ cannot coalesce with $(a,b)$ in population $i$ unless $a$ and $b$ have coalesced either in population $i$ or one of its descendant populations in the species tree.

A coalescent history can be viewed as an event under the coalescent model. Moreover, $H_{\psi,T}$ gives a partition of the event that the gene tree is $T$ into disjoint subevents $\mathbf h$. Although by definition $\PP(T,\mathbf h)=\PP(\mathbf h)$, for clarity we prefer to include the redundant reference to $T$ in this notation.
Note that $\PP_\sigma(T,\mathbf h)>0$ for every $\mathbf h\in H_{\psi,T}$ \citep{degnan2005}.

\begin{defn}\label{def:phist}
Let $\sigma = (\psi,\boldsymbol{\lambda})$ be a species tree on $X$, with $|X|=n$.  A \emph{population history} for $\psi$ is an $(n-1)$-tuple $\mathbf{y} = (y_i)_{i\in I_\psi}$ with $y_i\in \{0,1,\dots n-1\}$ satisfying

\begin{enumerate}
\item \label{cond:ph1} $\sum_{i\in I_\psi} y_i = n-1$,  and
\item \label{cond:ph2} $\sum_{j\in I_\psi} (1-y_j)\alpha_{ij} \ge 0$ for all $i \in {I_\psi}$.
\end{enumerate}
The set of all $(n-1)$-tuples satisfying conditions (1) and (2) is denoted $Y_\psi$.
\end{defn}

One should interpret a population history as indicating the number of coalescent events on each edge of a species tree that leads to a realization of some (unspecified) gene tree. Then
condition \eqref{cond:ph1} of the definition is interpreted as stating that over the full species tree all lineages ultimately coalesce into one --- i.e., there are a total of $n-1$ coalescences.

Condition \eqref{cond:ph2} requires more elucidation:
First note that for $i\in I_\psi$,
\begin{equation*}
\sum_{j\in I_\psi} \alpha_{ij}=\ell_i-1\end{equation*}
where $\ell_i$ is the number of leaf descendants of node $i$ on $\psi$.  This equivalence is due to the number of leaf descendants of a node of a binary tree being the number of internal descendants plus 1.  As an example, for the species tree in Figure 1(A), we have
$$\alpha_{11} = \alpha_{22} = \alpha_{31} = \alpha_{32} = \alpha_{33} = \alpha_{41} = \alpha_{42} = \alpha_{43} = \alpha_{44} = 1$$
and $\alpha_{ij} = 0$ for all other choices of $i$ and $j$.  To further illustrate the example,
$$\ell_3 = \alpha_{31} + \alpha_{32} = \alpha_{33} = 3 = 4 - 1.$$

Thus condition \eqref{cond:ph2} is equivalent to
$$\ell_i> \sum_{j\in I_\psi} y_j \alpha_{ij},$$
for each $i\in I_\psi$. This expresses that the number of coalescent events occurring on edges $i$ and below in the species tree cannot exceed the maximum possible for the lineages present in that part of the tree.

For any fixed species tree $\psi$ and gene tree $T$ there is a natural map $\Phi_{\psi,T}$ from the set of coalescent histories to the set of population histories, defined by `forgetting' which lineages coalesce: More formally
$$\Phi_{\psi,T}:H_{\psi,T}\to Y_\psi$$
$$\mathbf h=(h_j)_{j\in I_T} \mapsto \mathbf y=(y_i)_{i\in I_\psi},$$
where
$$y_i=\sum_{j\in I_T} \delta(h_j=i)$$
is the sum of indicators.

Population histories can also be viewed as events under the coalescent model.

\begin{defn} Given a species tree $\sigma$,
we say that a population history $\mathbf{y}\in Y_\psi$ is \emph{compatible} with a gene tree $T$ 
if they can be simultaneously realized, \emph{i.e.}, if $\PP_{\sigma}(T,\mathbf{y}) > 0$. 
We use $Y_{\psi,T}$ to denote the set of population histories compatible with a gene tree $T$.  
\end{defn}

Note that
 $Y_{\psi,T}=\Phi_{\psi,T}(H_{\psi,T})$, and $\PP_\sigma(T,\mathbf y)=\sum_{\mathbf h\in \Phi^{-1}_{\psi,T}(\mathbf y)}\PP_\sigma(T,\mathbf h).$

The loss of information in passing from coalescent histories to population histories is illustrated in Figure 1.  
In (A) and (B), two different coalescent histories for the matching gene tree yield the same population history.  For (A), the coalescent history is $\mathbf h=(4,3,4,4)$ because node 2 of the gene tree coalesces in population 3 of the species tree (hence $h_2=3$), while all other nodes coalesce in population 4 ($h_i=4$ for $i \ne 2$).  In (B), the coalescent history is $\mathbf h=(3,4,4,4)$ since node 1 of gene tree coalesces in population 3 of the species tree ($h_1=3$).  Both coalescent histories have one coalescence in population 3 and three coalescences in population 4, making their population histories both $(0,0,1,3)$.

Figure 1(C) and (D) illustrate another aspect of coalescent histories and population histories: that the probability of the same numbers of events in each species tree population can have a higher probability for a non-matching tree.
In (C), there are two coalescent events in population 3.  For this gene tree, either the $(a,b)$ coalescence or the $(c,d)$ coalescence can occur more recently within population 3 and result in the same gene tree topology, coalescent history, and population history, but different in-population rankings.  For the same population history with a caterpillar gene tree (D), however, the gene tree topology constrains the coalescence of lineage $c$ to be more ancient than the coalescence of $a$ with $b$.  This results in a lower probability for the same population history when the gene tree is a caterpillar compared to the matching gene tree.

\section{Results}

Our main result is the following:

\begin{thm}\label{T:MainTheorem} For a species tree $\sigma=(\psi,\lambda)$,
let  $T$ be any caterpillar gene tree, with $T\ne T_M$.  Then $\PP_{\sigma}(T) < \PP_{\sigma}(T_M)$. In particular, a caterpillar is never an AGT.
\end{thm}

As a consequence, we also obtain:

\begin{cor}There are no caterpillars in a wicked a forest.
\end{cor}

\begin{proof}
Any species tree in a wicked forest must have a topology which can be an AGT for some other species tree defined on the same taxa.  Since caterpillars cannot be AGTs by Theorem \ref{T:MainTheorem}, no species tree in a wicked forest can have a caterpillar topology.  
\end{proof}

Our proof of the theorem is built on a succession of lemmas. To simplify statements, we assume throughout that the species tree $\sigma=(\psi,\lambda)$ has been fixed.

The first lemma is immediately clear.

\begin{lem}\label{lem:pop_tot}
The probability of a gene tree $T$ can be written as $$ \PP_{\sigma}(T) =\sum_{\mathbf{y}\in Y_{\psi,T}} 
\PP_{\sigma}(T,\mathbf{y}).$$
\end{lem}

\begin{lem}\label{lem:john}
The matching gene tree $T_M$ is compatible with every population history. That is, $Y_{\psi,T_M}=Y_\psi$, so  $Y_{\psi,T}\subseteq Y_{\psi,T_M}$ for every gene tree $T$.
 \end{lem}

Though the proof of this is somewhat technical, the idea behind it is simple: With a population history $\mathbf y$ fixed, we pick any cherry on $T_M$, and have the coalescent event forming that cherry occur on the edge of $\psi$ as close to the leaves as possible among those allowed by $\mathbf y$. We then show that deleting the cherry from $T_M$ and $\psi$, and the coalescent event from $\mathbf y$ leads to trees and a population history with one fewer taxa, so an inductive argument gives the result.

\begin{proof}[Proof of Lemma \ref{lem:john}]
We must show that if
$\mathbf y\in Y_{\psi}$, then $\PP_{\sigma}(T_M,\mathbf y)>0$.
But
\begin{equation*}
\PP(T_M,\mathbf y)=\sum_{\mathbf h\in \Phi^{-1}_{\psi,T_M}(\mathbf y)} \PP (T_M,\mathbf h),
\end{equation*}
so it suffices to show there is some $\mathbf h\in H_{\psi,T_M}$ with $\Phi_{\psi,T_M}(\mathbf h)=\mathbf y$, since for such an $\mathbf h$,
$\PP_{\sigma}(T_M,\mathbf h)>0.$

We prove this by induction on the number of taxa $n$. The base case of $n=2$ is clear, since there is only one gene tree $T=T_M$,  and one coalescent history $\mathbf h$,  with $\PP_{\sigma}(T_M,\mathbf h)=1.$

Now assume the result is known for $n-1$ taxa. For $n$-taxon trees on taxa $X$,  identify the nodes of the matching gene tree $T_M$ with those of the species tree $\psi$ so that we may use the same notation to refer to either.
Pick an internal node $v$  on $\psi$ that is parental to exactly  two leaves, say $a$ and $b$. 
On both $\psi$ and $T_M$, prune the edge descending from the node  $v$ to leaf $a$, and then suppress that node, to obtain matching $(n-1)$-taxon trees $\tilde \psi$ and $\tilde T_M$ on taxa $X\smallsetminus\{a\}$. We may thus view the node sets of the four trees as satisfying $V_{\tilde \psi}=V_{\tilde T_M}\subset V_{\psi}=V_{T_M}$.

We similarly relate a population history $\mathbf y$ on $\psi$ to a population history $\tilde {\mathbf y}$ on $\tilde \psi$ in the following way: Let $w=w(\mathbf y)$ be the most recent vertex on $\psi$, ancestral or equal to population $v$, labeling a population in which a coalescent event occurs. That is
$$w=\min\{ i \mid \alpha_{iv}=1, \ y_i>0\},$$
where the minimum is taken with respect to the ancestral relationship.
Then let $\tilde {\mathbf y}=(\tilde y_i)_{i\in I_{\tilde \psi}}$ where $\tilde y_i= y_i-\delta(i=w)$.
(In essence, this simply removes one coalescent event from the population above $w$, but in the case where $w=v$, so $y_v=1$, this is done through dropping $y_v$ from $\tilde {\mathbf y}$.)  
 \medskip
 
We next verify that $\tilde {\mathbf y}$ is a population history for $\tilde \psi$, by showing it satisfies the appropriate constraints. Clearly it has non-negative entries.
That condition \eqref{cond:ph1} of Definition \ref{def:phist} is satisfied for $\tilde {\mathbf y}$ is also clear, since we have reduced the sum of the entries in $\mathbf  y$ by 1.

For the inequality constraints of condition \eqref{cond:ph2}, first suppose $w=v$.
Then for all $j\in I_{\tilde\psi}$ we have $\tilde y_{j}=y_{j}$. Thus for $i\in I_{\tilde \psi}$,
\begin{align*}
\sum_{j\in I_{\tilde \psi}}(1-\tilde y_{j})\alpha_{ij} &= \sum_{j\in I_{\tilde \psi}} (1- y_{j}) \alpha_{ij}\\
&= -(1-y_v)\alpha_{iv}+\sum_{j\in I_{ \psi}}(1- y_{j}) \alpha_{ij}.
\end{align*}
But the first term in this last expression is 0, since $y_v=1$, and the second is non-negative because $\mathbf y$ is a population history vector for $\psi$. Thus condition \eqref{cond:ph2} is established in this case.

Now suppose $w$ is ancestral to $v$. Then $\tilde y_{i}=y_{i}$ for all $i\ne w$, while $\tilde y_w=y_w-1$, so
\begin{align*}
\sum_{j\in I_{\tilde \psi}}(1-\tilde y_{j}) \alpha_{ij} &= \alpha_{iw}+ \sum_{j\in I_{\tilde \psi}}(1- y_{j}) \alpha_{ij}\\
&=\alpha_{iw} -(1-y_v)\alpha_{iv}+\sum_{j\in I_{ \psi}}(1- y_{j}) \alpha_{ij}.
\end{align*}
By the minimality of  $w$, we know $y_v=0$. It will follow that the above expression is non-negative in any case when $\alpha_{iw}-\alpha_{iv}\ge 0$.
This is true if either $i$ is ancestral or equal to $w$ (and hence ancestral to v), or $i$ is not ancestral to v (and hence not ancestral to $w$). 

The remaining subcase to consider is when $i$ is ancestral to $v$ but not ancestral or equal to $w$, i.e., $i$ lies between $v$ and $w$. In this situation $\alpha_{iw}-\alpha_{iv}=-1$, so we must show
$$\sum_{j\in I_{ \psi}}(1- y_{j}) \alpha_{ij}\ge 1.$$
But if $i$ has two internal nodes as children, say $k$ and $l$, then
$$\sum_{j\in I_{ \psi}}(1- y_{j}) \alpha_{ij}=(1-y_{i}) +\sum_{j\in I_{\psi}}(1- y_{j}) \alpha_{kj} +\sum_{j\in I_\psi}(1- y_{j}) \alpha_{lj}.$$
The two sums on the right are non-negative, because $\mathbf y$ is a population history vector for $\psi$. Since $y_{i}=0$ by the minimality of $w$, we obtain the needed inequality. The case where $i$ has only one internal node as a child is similar.

This concludes the argument that $\tilde{\mathbf y}$ is a valid population history for $\tilde \psi$.

\smallskip

Since $\tilde{\mathbf y}$ is a population history for $\tilde \psi$, by the induction hypotheses there is a coalescent history $\tilde{\mathbf h}\in H_{\tilde\psi,\tilde T_M}$
with $\Phi_{\tilde \psi,\tilde T_M}(\tilde{\mathbf h})=\tilde{\mathbf y}$.
Define a coalescent history for $T_M$ on $\psi$ by $\mathbf h =(h_i)_{i\in I_\psi}$ with
$h_i= \tilde h_i$ for $i\in \tilde \psi$ and $h_v=w$.

To verify that $\mathbf h\in H_{\psi,T_M}$, we must check that it satisfies the constraints of Definition \ref{def:chist}. 
For a matching tree, condition \eqref{cond:ch1} is equivalent to saying that $h_i$ must be ancestral or equal to $i$. For $i\ne v$, this follows immediately from the fact that $\tilde h_i$ is ancestral to $i$ on $\tilde T_M$. Since $h_v=w$, and $w$ ancestral to $v$,  the constraint is satisfied in all cases.

For condition \eqref{cond:ch2}, we must check that
if $i$ is ancestral to $j$ on $T_M$, then $h_i$ is ancestral or equal to $h_j$ on $\psi$.  For $j\ne v$, this follows immediately from the analogous property for $\tilde h$, but for $j=v$ requires more explanation.

Suppose $i$ is ancestral to $v$ on $T_M$.  
Then $i$ is ancestral to leaf $b$ on $T_M$, hence $i$ is ancestral to leaf $b$ on $\tilde T_M$, so $\tilde h_{i}$ must be ancestral to leaf $b$ on $\tilde \psi$, so $h_i$ is ancestral to leaf $b$ on $\psi$, and hence ancestral or equal to node $v$.
If $w=v$ so $h_v=v$, we are done verifying the constraint. If $w$ is ancestral to $v$, then since $\Phi_{\psi, \tilde T_M}(\tilde {\mathbf h}) =\tilde{\mathbf y}$, the minimality of $w$ ensures no entries of $\tilde {\mathbf h}$ are nodes between leaf $b$ and node $w$ on $\tilde \psi$. Thus $h_i$ does not lie between leaf $b$ and node $w$ on $\psi$. Since
$h_i$ is ancestral to $v$, it must therefore be ancestral or equal to $w=h_v$. 

Finally, observing $\Phi_{\psi,T_M}(\mathbf h)=\mathbf y$ completes the proof.
\end{proof}

\begin{lem}\label{lem:numerator}
If $T$  is a caterpillar gene tree, then for any population history $\mathbf{y}\in Y_\psi$, $\PP_{\sigma}(T,\mathbf{y}) \le \PP_{\sigma}(T_M,\mathbf{y})$.
\end{lem}

\begin{proof}
From Lemma \ref{lem:john}, if $\PP_{\sigma}(T,\mathbf{y})>0$, then $\PP_{\sigma}(T_M,\mathbf{y})>0$ as well, so $T_M$ can be realized with $\mathbf y$. 

Now each of these probabilities can be expressed as a product of two terms: one which depends only on $T$ and $\mathbf y$, and one which  depends only on $\mathbf y$ and $\boldsymbol \lambda$ . More specifically,
\begin{align}\PP_{\sigma}(T,\mathbf{y})&=R_{T,\mathbf y} f(\mathbf y,\boldsymbol \lambda),\notag\\
\PP_{\sigma}(T_M,\mathbf{y})&=R_{T_M,\mathbf y} f(\mathbf y,\boldsymbol \lambda),\label{eq:factor}
\end{align}
where $R_{T,\mathbf y}$ counts the number of coalescent histories with in-population rankings consistent with the gene tree $T$ and $\mathbf y$, and
$$f(\mathbf y,\boldsymbol \lambda)=\prod_{i=1}^{n-1}
\frac 1{d_{j_ik_i}}
 g_{j_ik_i}(\lambda_i),$$ with $j_i=\ell_i-(\sum_{j\ne i} y_j \alpha_{ij})$ the number of lineages `entering' population $i$ from below 
and $k_i=j_i-y_i$ the number of lineages `leaving' population $i$ above, $d_{jk}$ the number of sequences of coalescent events that may occur for $j$ labeled entering lineages to coalesce to $k$ leaving lineages, and $g_{jk}(u)$ is the function which gives the probabilities that $j$ lineages in a population coalesce to $k$ lineages in $u$ coalescent units \citep{degnan2005,rosenberg2003,tavare1984,wakeley2008}.

Since $T$ is a caterpillar, its realization requires a specific ranked ordering to coalescent events, so $R_{T,\mathbf y}=1$.
Since $R_{T_M,\mathbf y}\ge 1$ the lemma follows from equations \eqref{eq:factor}.
\end{proof}

\begin{lem}\label{lem:minpophist}
The population history $\mathbf  1=(y_i)_{i\in I_\psi}$  with all $y_i=1$ is consistent with the matching gene tree $T_M$, but no other.
That is,
$\mathbf 1\in Y_{\psi,T}$ if and only if $T=T_M$.
\end{lem}

\begin{proof}
That $\mathbf 1\in Y_{\psi,T_M}$ is a consequence of Lemma \ref{lem:john}.

Suppose $\mathbf 1\in Y_{\psi,T}$. To establish that $T=T_M$ it is enough to show these gene trees must have the same clades \citep{semple2003}.
 Since $\PP_\sigma(T,\mathbf 1)>0$, $T$ is realizable with one coalescent event on each internal edge of $\psi$. But for any $i$, there are $\ell_i$ taxa descended from node $i$ of $\psi$, and
population history $\mathbf 1$ implies $\ell_i-1$ coalescent events occur on or below the edge above $i$. Thus for both $T$ and
$\mathbf y=\mathbf 1$ to be simultaneously realized,  the lineages of all taxa descended from  $i$ on $\psi$ must coalesce to form a clade on $T$. Thus every clade of $\psi$ is a clade on $T$, so $T=T_M$.
\end{proof}


\medskip

\begin{proof}[Proof of Theorem \ref{T:MainTheorem}]  
From the lemmas,
\begin{align*}
\PP_\sigma(T) &= \sum_{\mathbf y \in Y_{\psi,T} }\PP_\sigma (T,\mathbf y)
\le \sum_{\mathbf y \in Y_{\psi,T}} \PP_\sigma(T_M,\mathbf y)
< \sum_{ \mathbf y  \in Y_{\psi,T_M}} \PP_\sigma(T_M,\mathbf y)
= \PP_\sigma(T_M).
\end{align*}
The first equality is from Lemma \ref{lem:pop_tot}; the next inequality from Lemma \ref{lem:numerator}; the next from Lemmas \ref{lem:john} and
\ref{lem:minpophist}; and the final equality from Lemma \ref{lem:pop_tot} again.

\end{proof}

\begin{rem} \label{rem:hist}
Let $\sigma=(\psi,\lambda)$ be a species tree on taxa $X$, and let $T$ be any nonmatching caterpillar gene tree on $X$.   
Then the above considerations show
\begin{equation}\label{E:nonmatching}
|H_{\psi,T}| = |Y_{\psi,T}| < |Y_{\psi,T_M}| \le |H_{\psi,T_M}|,
\end{equation}
i.e,  the number of consistent coalescent histories is larger for matching trees than for any nonmatching caterpillar tree.  It has previously been shown for some species trees  that the number of coalescent histories can be larger for 
a nonmatching, noncaterpillar tree than for a matching tree, although the smallest trees for which this occurs have 7 taxa \citep{rosenberg2010}.    Equation \eqref{E:nonmatching} shows that gene trees with more coalescent histories than the matching tree are never caterpillars, which presents a combinatorial analog to the result that caterpillar gene trees can never be AGTs. 
\end{rem}

\section{ Anomalous Ranked and Unrooted Gene Trees}

Recently, the concept of anomalous gene trees has also been extended to ranked gene trees \citep{degnan2012b,disanto2014} and unrooted gene trees \citep{degnan2013}.  

A ranked gene tree topology encodes the relative timing of the branches, so that, as an example, the ranked gene tree topologies in Figure 1(A) and (B) are distinct because the ordering of the $(a,b)$ and $(c,d)$ coalescences are reversed, even though the unranked gene tree topologies are the same.  An {\it anomalous ranked gene tree} (ARGT) is a ranked gene tree that is more probable than the ranked gene tree that matches the ranked species tree \citep{degnan2012b}.  The ranked gene tree in Figure 1(A) matches the ranked species tree, while the ranked gene tree in Figure 1(B) does not.  In spite of the results of this paper, caterpillar gene trees {\it can} be ARGTs \citep{degnan2012b}, i.e., a caterpillar gene tree can be more probable than a matching ranked gene tree, even though it must be less probable than the matching unranked gene tree.

On the other hand, neither caterpillar nor pseudo-caterpillar species trees have ARGTs \citep{degnan2012b}. (A pseudo-caterpillar tree is one obtained from a caterpillar by attaching two edges to each leaf in the caterpillar's cherry.  The species trees in Figure 1 are pseudo-caterpillars.) Therefore, extending the concept of a wicked forest to ranked gene trees, there are no caterpillars or pseudocaterpillars in a {\it wicked forest for ranked gene trees} (a nonempty set $W$ of distinct species trees where the ranked topology each member is an ARGT for all other members).   An example of a wicked forest for ranked gene trees using 10 taxa is shown in Figure 3.  The smallest number of taxa needed for a wicked ranked forest is unknown.

\begin{figure*}
\begin{center}
\includegraphics[width=.48\textwidth,angle=270]{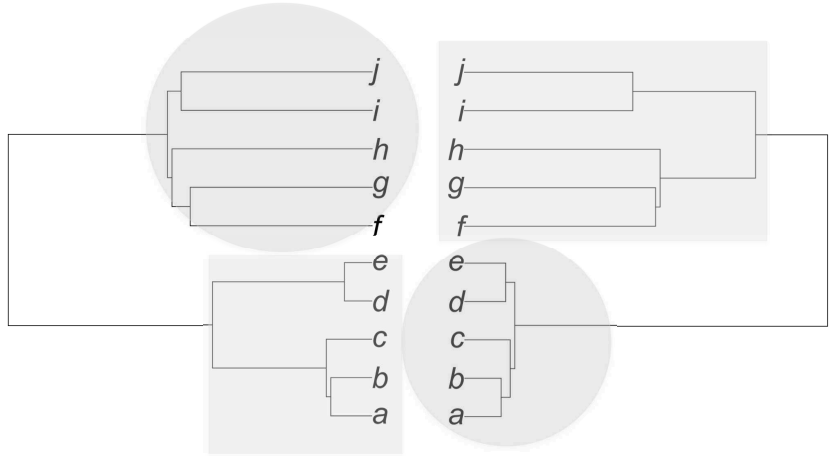}
\caption{A ranked wicked forest.  The two trees have the same unranked topology but have different rankings since, for example, $(d,e)$ is the most recent common ancestor for the left tree, while $(a,b)$ has the most recent common ancestor in the right tree.  Subtrees in rectangular shaded boxes have ARGTs, shown in corresponding circular shaded regions in the facing tree. For the subtrees in rectangular regions, there are relatively long branches separating the two- and three-taxon clades. Subtrees that are not in boxes have low probability for any particular sequence of coalescences because all branches are short.  These subtrees have short branches separating two- and three-taxon clades.}
\end{center}
\end{figure*}

The gene tree probabilities for Figure 3 are most easily approximated by assuming that the branches between the root and the shaded regions are very long, so that coalescence of all available lineages is virtually guaranteed on these branches.  Then probabilities for the left and right shaded subtrees can be obtained using formulas from \citet{degnan2012}. 
For the tree on the left in Figure 3, let the subtree in the rectangular box be $\sigma_{  \text{  \scalebox{0.6}{$\square$}}}^1$, and the tree in the circular shaded region be $\sigma_{\circ}^1$, so that the overall species tree is $\sigma_L = (\sigma_{  \text{  \scalebox{0.6}{$\square$}}}^1{:}
\lambda_1,\sigma_{\circ}^1{:}\lambda_2)$, where $\lambda_1$ and $\lambda_2$ are very large.  Similarly, the tree on the right of Figure 3 is $ \sigma_R = (\sigma_{  \text{  \scalebox{0.6}{$\square$}}}^2{:}
\lambda_3,\sigma_{\circ}^2{:}\lambda_4)$.  Here $\sigma_{  \text{  \scalebox{0.6}{$\square$}}}^1$ and $\sigma_{\circ}^2$ are species trees on $X_1 = \{a,b,c,d,e\}$ and $\sigma_{\circ}^1$ and $\sigma_{  \text{  \scalebox{0.5}{$\square$}}}^2$ are species trees on $X_2 = \{f,g,h,i,j\}$.  We let $T_{  \text{  \scalebox{0.6}{$\square$}}}^i$ and $T_{\circ}^i$ be the matching ranked gene trees for $\sigma_{  \text{  \scalebox{0.6}{$\square$}}}^i$ and $\sigma_{\circ}^i$, respectively.    Let $T_L$ and $T_R$ denote the matching ranked gene trees for the left and right trees, respectively.

From \citet{degnan2012}, branch lengths can be chosen so that  if  $\sigma_L$ is the species tree, then with probability arbitrarily close to 2/8, the ranked gene tree restricted to taxa $X_1$ is $T_{  \text{  \scalebox{0.6}{$\square$}}}^1$, and with probability arbitrarily close to 3/8 is $T_{\circ}^2$.    Branch lengths can also be chosen so that for taxa $X_2$, the ranked gene tree restricted to taxa $X_2$ has nearly equal probability of being either $T_{\circ}^1$ or $T_{  \text{  \scalebox{0.6}{$\square$}}}^2$.  Therefore, for some choices of branch lengths, 
$$\frac{\mathbb P_{\sigma_L}(T_R)}{\mathbb P_{\sigma_L}(T_L)} \approx \frac{3}{2}.$$
Similar arguments show that $T_L$ can be approximately 1.5 times as probable as $T_R$ when $\sigma_R$ is the species tree. In this example, the wicked forest contains two species trees with identical unranked topologies but different ranked topologies.
Examples of wicked forests for ranked gene trees that contains trees with different topologies can also be constructed.  For example, one could swap taxa $b$ and $c$ in $\sigma_L$ but not $\sigma_R$ and still obtain a wicked forest for ranked gene trees.

\smallskip

Probabilities of unrooted gene trees can be obtained by summing over the probabilities of all rooted gene trees with the same unrooted topology.  An unrooted caterpillar tree is a binary tree where every internal node is connected by an edge to a leaf node.
Unrooted caterpillar gene trees can be  anomalous unrooted gene trees (AUGTs), i.e., more probable than the unrooted gene tree with the same unrooted topology as that of the species tree.  Figure \ref{fg:Uwicked} shows a {\it wicked forest for unrooted gene trees}, which we define as a nonempty set $W$ of rooted species trees such that for $\sigma_i, \sigma_j \in W$, $P_{\sigma_i}(u(T_j)) > P_{\sigma_i}(u(T_i))$ for $i \ne j$, where $u(T_i)$ is the unrooted topology of $T_i$, and $T_i$ has the same rooted topology and $\sigma_i$. This example shows that caterpillars can be in a wicked forest for unrooted trees.

\begin{figure*}
\begin{center}
\includegraphics[width=.48\textwidth,angle=270]{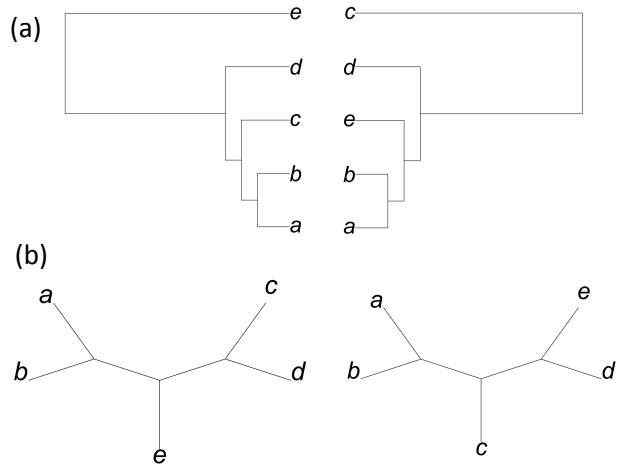}
\caption{(a) Two rooted caterpillar species trees constituting a wicked forest for unrooted gene trees.  For these species trees, the two shorter internal branches have branch length 0.05 coalescent units while the longer one has length 0.5 coalescent units. (b) The unrooted gene tree on the left is the most probable unrooted gene tree given the species tree on the left in (a).  The unrooted gene tree on the right in (b) is the most probable unrooted gene tree given the species tree on the right in (a).}\label{fg:Uwicked}
\end{center}
\end{figure*}

\section{Future work on AGTs}

The fact that caterpillar gene trees cannot be AGTs  fits the intuition that AGTs are more easily found among gene trees with more balanced topology than the species tree \citep{degnan2013,rosenberg2013b}.  For
unbalanced species trees, choosing sufficiently short branch lengths gives gene trees with a higher amount of tree balance greater probability \citep{degnan2006}.  However, the fact that even perfectly balanced species trees can have AGTs \citep{degnan2006} suggests that it is difficult to characterize all AGTs.  Thus, there is still an open question:  for a given species tree topology, which gene tree topologies can be AGTs?

The strategy of Degnan~(2013) can be used to predict many of the AGTs for a given species tree: First one considers a smaller species tree induced by taking a subset of taxa.  If this smaller tree has AGTs, then ones for the larger tree can be predicted by re-grafting the removed taxa onto the smaller AGTs. 
As an example, for the  species tree  $(((a,b),(c,d)),e)$, called a {\it pseudo-caterpillar} by \citet{rosenberg2007},
removing taxon  $c$ results in the caterpillar $(((a,b),d),e)$, which can have AGTs $((a,b),(d,e))$, $((a,d),(b,e))$, and $((a,e),(b,d))$.  Placing $c$ back on these AGTs results in $((a,b),((c,d),e)$, $((a,(c,d)),(b,e))$ and $((a,e),(b,(c,d)))$. While this perhaps suggests that the 5-taxon pseudo-caterpillar species tree cannot have a pseudo-caterpillar AGT, the verification of that fact currently depends on a detailed calculation of gene tree probabilities \citep{rosenberg2008}.  

While it would be desirable to have an efficient way of determining which topologies can be AGTs for a given species tree, potentially more valuable would be methods for determining the set of species trees for which a given gene tree can be most probable.  Such candidate species trees could then be used to reduce the search space for the optimal species trees to explain a set of gene trees \citep{fan2011}. 

\medskip

Further results on AGTs may also be helpful in interpreting results of species trees inference by concatenation of gene sequences.  In particularly, simulations \citep{kubatko2007,degdeg2010} as well as theoretical results \citep{roch2015} have shown that when maximum likelihood is used to infer a tree based on concatenated DNA sequences, the inferred tree can be misleading, in the sense that concatenating more genes  can be more likely to lead to an erroneous inferred species tree.  In simulations where concatenation has been misleading, the returned tree is often an AGT.  Simulations also suggest that concatenation performs better when the true species tree is balanced \citep{leache2011}, and thus AGTs are less common \citep{degnan2006,rosenberg2008,degnan2013}. Studies are needed to determine whether in larger trees inferred from empirical data, certain tree shapes inferred from concatenation tend to be more reliable than others.

\section*{Acknowledgments}

This research was begun with support for the second author provided by an Erskine Fellowship at the University of Canterbury.

\bibliography{bibfile}

\end{document}